\documentclass[doublecol]{epl2} 
\usepackage{hyperref}
\usepackage[mathscr]{eucal}
\usepackage{amsmath,amsfonts,amssymb,amsthm}

\newcommand\rmd{\mathrm{d}}
\newtheorem{theorem}{Theorem}
\newtheorem{corollary}{Corollary}

\title{Integrability of certain Hamiltonian systems in $2D$ variable curvature spaces}
\shorttitle{Integrability of certain Hamiltonian systems in $2D$ variable curvature spaces} 

\author{W. Szumi\'{n}ski\inst{1} \and A.~A. Elmandouh\inst{2,3}}
\shortauthor{W. Szumi\'{n}ski  \etal}

\institute{                    
  \inst{1} Institute of Physics, University of Zielona G\'{o}ra Licealna 9, PL-65-407, Zielona G\'{o}ra, Poland \\
  \inst{2} Department of Mathematics and Statistics, College of Science, King Faisal University, P.O. Box 400, Al-Ahsa 31982, Saudi Arabia\\    \inst{3}  Department of Mathematics, Faculty of Science, Mansoura University, Mansoura 35516, Egypt}

\abstract{
The objective of this work is to examine the integrability of Hamiltonian systems in $2D$ spaces with variable curvature of certain types.  Based on the differential Galois theory, we announce the necessary conditions of the integrability.  They are given in terms of arithmetic restrictions on values of the parameters describing the system.  We apply the obtained results to some examples to illustrate that the applicability of the obtained result is easy and effective.   Certain new integrable examples are given.  The findings highlight the applicability of the differential Galois approach in studying the integrability of Hamiltonian systems in curved spaces, expanding our understanding of nonlinear dynamics and its potential applications. \vskip 10pt
\textbf{Declaration}:  The article has been published in~\cite{Szuminski_Elmandouh_2025}, and the final version is available at: \textbf{\href{ https://doi.org/10.1209/0295-5075/adc150}{ https://doi.org/10.1209/0295-5075/adc150} }
}
\begin{document}

\maketitle

\section{Introduction}
The study of the integrability of Hamiltonian systems is one of the central problems in mathematical physics. An $n$-dimensional Hamiltonian system is said to be integrable in the Liouville sense if it possesses $n$ functionally independent constants of motion that are in involution, i.e., their Poisson brackets vanish~\cite{A1}. Integrable Hamiltonian systems exhibit several advantages: their behavior can be analyzed globally over an infinite interval of time, offering profound insights into the long-term evolution of the system. This global analysis is particularly valuable as it enables the identification of invariant tori and other persistent structures in phase space.

Furthermore, integrable systems serve as a foundation for studying non-integrable Hamiltonian systems through perturbation theories. By examining deviations from integrable behavior, one can gain valuable insights into stability and the onset of chaos in non-integrable systems. Another significant property of integrable systems is that their explicit solutions can be obtained by quadratures, meaning that the generalized coordinates can be expressed as functions of time via integrals~\cite{A2}.

Integrable systems have numerous applications in physical sciences. In celestial mechanics, they play a crucial role in understanding the motion of planets, satellites, and other celestial bodies~\cite{A2b,A2c}. In quantum mechanics, integrable systems contribute to the solvability of certain quantum models~\cite{A2d,A2e}. Moreover, they are instrumental in uncovering the geometric structure of phase space and elucidating the underlying symmetries of physical laws. Their importance in theoretical physics lies in their ability to provide a deeper understanding of the mathematical and geometric foundations of dynamical systems~\cite{A2f,A2g}.

\par In physics, many mechanical systems can be described as two-dimensional systems. These systems are generally described by the following Lagrangian

\begin{equation}
L=\frac{1}{2}(a_{11}\dot{\xi}^2+2a_{12}\dot{\xi}\dot{\eta}+a_{22}\dot{\eta}^2)-V(\xi, \eta),\label{IN1}
\end{equation}
where $a_{ij}$ are functions relying on $\xi, \eta$ and $V(\xi,\eta)$ refers to the potential function. According to Birkhoff's theorem \cite{N1}, one can always utilize isometric coordinates (say) $q_1, q_2$ to transform the Lagrangian~(\ref{IN1}) to
\begin{equation}
L=\frac{1}{2\Lambda}(\dot{q}_1^2+\dot{q}_2^2)-V(q_1,q_2),\label{IN2}
\end{equation}
where $\Lambda$ is a function in $q_1,q_2$.
The Hamiltonian corresponding to the Lagrangian~(\ref{IN2}) takes the form
\begin{equation}
H=\frac{\Lambda}{2}(p_1^2+p_2^2)+V(q_1,q_2).\label{IN3}
\end{equation}
Hamiltonian~(\ref{IN3}) describes the motion of a particle under the influence of potential forces derived by the potential function, and this motion happens on a surface with Gaussian curvature, given by 
\begin{equation}
\kappa=\frac{\Lambda}{2}\left(\frac{\partial^2\Lambda}{\partial q_1^2}+\frac{\partial^2\Lambda}{\partial q_2^2}\right).\label{IN4}
\end{equation}
\par The majority of the integrability studies related to the Lagrangian~(\ref{IN2}), or equivalently to the Hamiltonian~(\ref{IN3}), are divided into two categories. The first one is finding sufficient conditions for integrability, that is, enumerating the required number of constants of motion. Many methods are used, such as the Darboux and Yehia methods, see,e.g., \cite{A3,A4,A5,A6,A7,A8,A11,A13,A14}. All these methods are interested in searching for the first integrals of motion when the configuration space is Euclidean. Most of this work was listed in the Hietarienta review up to 1987. Yehia's proposed method works whether the configuration space is Euclidean or not. Hence, it was applied to construct new integrable problems in a rigid body dynamic and its generalization to a gyrostat. The second one endeavors to get the necessary conditions for the integrability of a considered model.  These conditions can be obtained by applying numerous methods such as Painlevé property \cite{A17} or Ziglin approach \cite{A18}.
It was recently noticed the existence of a certain connection between Liouville integrable Hamiltonian systems and the properties of the differential Galois group $\mathcal{G}$ of the variational equation for these systems. Because $\mathcal{G}$ is an algebraic group, it was written as a union of a finite number of connected components in which one of them contains the identity, and it is called identity components of $\mathcal{G}$ and it is denoted by $\mathcal{G}_0$. The necessary conditions for the integrability were presented as a restriction on the identity component $\mathcal{G}_0$. The Morales-Ruiz and Ramis theorem states the following theorem.
\begin{theorem}[Morales-Ruiz and Ramis~\cite{N3}]
Let us assume a Hamiltonian system is Liouville integrable in a neighborhood of a phase curve $\Gamma$ corresponding to a certain particular solution. Then, the identity component $\mathcal{G}_{0}$ of the differential Galois group $\mathcal{G}$ associated with the variational equations along $\Gamma$ is abelian.
\end{theorem}
The integrability study of certain forms of the Hamiltonian~(\ref{IN3}) is based on applying the differential Galois theory. The necessary conditions are obtained as those that guarantee the solvability of the differential Galois group associated with the normal variational equation along a particular solution. We classify and present the previous works according to the Gaussian curvature.

\begin{enumerate}
\item
\emph{Hamiltonian systems with zero Gaussian curvature}. In physics, these systems describe the motion of a particle in the flat Euclidean plane. The integrability of the  Hamiltonian
\begin{equation}
H=\frac{1}{2}(p_1^2+p_2^2)+V(q_1, q_2),\label{IN5}
\end{equation}
with homogeneous potential of degree $k\neq0$ was studied by Morales-Ruiz and Ramis in \cite{N4}. This study was completed by Casale et al.\cite{N5} when $k=0$. For the non-homogeneous potentials that can be decomposed in a unique way into homogeneous parts, the integrability of this potential implies the integrability of potentials of the lowest and highest homogeneity degrees, but this is usually not sufficient for integrability, and sometimes one can find additional obstructions  \cite{RE1}.

Szumin\'ski et al. in~\cite{N6} examine the integrability of certain classes of Hamiltonian systems in curved spaces. This Hamiltonian takes the form
\begin{equation}
H=\frac{1}{2}r^{m-k}\left(p_r^2+\frac{p_\theta^2}{r^2}\right)+r^mU(\theta),\label{IN6}
\end{equation}
where $m, k$ are integers and $U$ is a complex meromorphic function.  Hamiltonian~(\ref{IN6}) characterizes the motion of a particle in the flat Euclidean plane. In~\cite{N7}, Elmandouh examine the integrability of a Hamiltonian in the form
\begin{equation}
H=\frac{1}{2}\left(p_r^2+\frac{p_\theta^2}{r^2}\right)+r^m+r^{m+k}U(\theta),\label{IN7}
\end{equation}
where $m, k$ are integers whereas $U$ is a complex meromorphic function.
\item  \emph{Hamiltonian system with constant Gaussian curvature}. These systems are interpreted physically as the motion either on a standard sphere or pseudo-sphere, according to the Gaussian curvature, is either positive or negative constant.   In \cite{N8}, Maciejewski et al. examined the integrability of the Hamiltonian taking the form
\begin{equation}
H=\frac{1}{2}\left(p_r^2+\frac{p_\theta^2}{S_\kappa^2(r)}\right)+S_\kappa^k(r)U(\theta),\label{IN8}
\end{equation}
where $U$ is a meromorphic function and $S_\kappa(r)$ is defined as
\begin{eqnarray}
S_{\kappa}(r):=\left\{\begin{array}{cc}
                         \frac{1}{\sqrt{\kappa}}\sin\sqrt{\kappa}r &  \qquad \text{if}\quad\kappa>0 \\
                         r & \qquad \text{if} \quad \kappa=0 \\
                         \frac{1}{\sqrt{-\kappa}}\sinh\sqrt{-\kappa}r& \qquad \text{if}\quad\kappa<0
                       \end{array}\right.
\end{eqnarray}
\item  \emph{Hamiltonian systems with variable Gaussian curvature}. Such Hamiltonian characterizes the motion of a particle on a surface with variable Gaussian curvature. Elmandouh in~\cite{Elmandouh:18::}, and Szumiński in~\cite{Szuminski:21::} investigated the integrability of the Hamiltonian
\begin{equation}
H=\frac{1}{2}r^n\Lambda(\theta)\left(p_r^2+\frac{p_\theta^2}{r^2}\right)+r^m U(\theta),\label{IN9}
\end{equation}
where $n, m$ are integers and $\Lambda(\theta), U$ are meromorphic functions.

\end{enumerate}
The current work is interested in studying the integrability of Hamiltonians, which are generalizations of these belonging to the third type. Namely, we consider a Hamiltonian of the following general form
\begin{equation}
H=\frac{1}{2}(r^n\Lambda_{2}(\theta)+\Lambda_1(\theta))\left(p_r^2+\frac{p_\theta^2}{r^2}\right)+r^mU(\theta),\label{Eq1}
\end{equation}
where $m, n$ are integers and $\Lambda_1,\Lambda_{2}, U$ are  meromorphic functions. In physics, Hamiltonian~\eqref{Eq1} describes the motion of a particle due to the effect of potential forces derived from the potential $r^mU$, and this motion takes plane on a surface with Gaussian curvature in the form
\begin{equation}
\kappa=\frac{1}{2r^2}\left[r^n\Lambda_{2}^{\prime\prime}+\Lambda_{1}^{\prime\prime}\right]+\frac{n^2r^n\Lambda_{1}\Lambda_{2}
-[r^n\Lambda_{2}^\prime+\Lambda_{1}^\prime]^2}{2r^2[r^n\Lambda_{2}+\Lambda_{1}]},\label{Gauss}
\end{equation}
where dash indicates the derivative concerning $\theta$. Notice the Hamiltonian \eqref{Eq1} represents a general class of Hamiltonian systems introduced above. For instance, it turns into the Hamiltonian \eqref{IN9} if $\Lambda_{1}=0$ and it is transformed to the Hamiltonian \eqref{IN6} if $\Lambda_{2}=1, \Lambda_{1}=0$.
\section{Variational equations and the integrability}
The Hamilton equations corresponding to the Hamiltonian~\eqref{Eq1} are as follows
\begin{equation}
\begin{split}
\dot{r}&=[r^n\Lambda_{2}+\Lambda_{1}]p_r,\\
\dot{p}_r&=-\frac{nr^{n-1}}{2}\Lambda_{2}\left(p_r^2+\frac{p_\theta^2}{r^2}\right)+[r^n\Lambda_{2}+\Lambda_{1}]\frac{p_\theta^2}{r^3}-mr^{m-1}U,\\
\dot{\theta}&=\frac{1}{r^2}[r^n\Lambda_{2}+\Lambda_{1}]p_\theta,\\
\dot{p}_\theta&=-\frac{1}{2}[r^n\Lambda_{2}^\prime +\Lambda_{1}^\prime ]\left(p_r^2+\frac{p_\theta^2}{r^2}\right)+r^mU^\prime .\label{Eq2}
\end{split}
\end{equation}
The Hamiltonian Eqs. \eqref{Eq2} admit the energy integral
\begin{equation}
\frac{1}{2}(r^n\Lambda_{2}+\Lambda_{1})\left(p_r^2+\frac{p_\theta^2}{r^2}\right)+r^mU=h,\label{Eq3}
\end{equation}
where $h$ is the energy constant. If $\Lambda_{2}^{\prime}(c)=\Lambda_{1}^{\prime}(c)=U^\prime(c)=0$ for a specific $c\in\mathbb{C}$, then the Hamilton Eqs.~\eqref{Eq2} possess a $2D$ invariant manifold in the form
\begin{equation}
\mathcal{S}=\left\{(r, p_r, \theta, p_\theta)\in\mathbb{C}^4 \ \big{|}\ \theta=c, \, p_\theta=0\right\}.\label{Eq4}
\end{equation}
The phase curves parameterized by the energy $h$ foliate
the invariant manifold $\mathcal{S}$ and for $(r, p_r, \theta, p_\theta)\in\mathcal{S}$, the energy integral \eqref{Eq3} takes the form
\begin{equation}
\frac{1}{2}[r^n\Lambda_{2c}+\Lambda_{1c}]p_{r}^{2}+r^mU_c=h,\label{Eq5}
\end{equation}
where for the short notation, we defined
\begin{equation}
\label{eq:dd}
\Lambda_{1c}:=\Lambda_1(c),\quad \Lambda_{2c}:=\Lambda_2(c),\quad U_c:=U(c).
\end{equation}
Let $\mathbf{X}=[R, P_R, \Theta, P_\Theta]^T$ indicate the variations of the variables $\mathbf{x}=[r, p_r, \theta, p_\theta]^T$, then the first-order variational equations along the particular solution that lies on $\mathcal{S}$ admit the form
\begin{equation}
\label{Eq6}
\mathbf{\dot  X}=\mathbf{M}\cdot\mathbf{X}.
\end{equation}
Here $\mathbf{M}$ is the Jacobian matrix
\begin{eqnarray}
\left(
                \begin{array}{cccc}
                  nr^{n-1}\Lambda_{2c}p_r & r^n\Lambda_{2c}+\Lambda_{1c} & 0 & 0 \\
                 A  & 0 & 0 & 0 \\
                  0 & 0 & 0 & \frac{r^n\Lambda_{2c}+\Lambda_{1c}}{r^2} \\
                  0 & 0 & B  & 0 \\
                \end{array}
              \right),           \label{Eq66}
\end{eqnarray} and $A, B$ are given by
\begin{equation}
\begin{split}
&A(r,p_r)=\frac{n}{2}(1-n)r^{n-2}\Lambda_{2c}p_r^2+m(1-m)r^{m-2}U_c,\\
&B(r,p_r)=-\frac{1}{2}\left(\Lambda_{1c}''+r^n \Lambda_{2c}''\right)p_r^2-r^m U_c''.\label{Eq7}
\end{split}
\end{equation}
Since the particular solution corresponds to the motion  on $(r, p_r)$ plane,
the normal variational equations are provided by the following  closed subsystem
\begin{eqnarray}
\dot{\Theta}=\frac{r^n\Lambda_{2c}+\Lambda_{1c}}{r^2}P_\Theta,\qquad
\dot{P}_\Theta=B\,\Theta.\label{Eq8}
\end{eqnarray}
System \eqref{Eq8} can be rewritten as a single second-order differential equation in the form
\begin{eqnarray}
\ddot \Theta+a(r,p_r) \,\dot \Theta+b(r,p_r)\,\Theta=0,
\label{Eq9}
\end{eqnarray}
where
\begin{equation*}
\begin{split}
 a(r,p_r)&=\frac{\Lambda_{1c}+(2-n)r^n\Lambda_{2c}}{r}p_r,\\
 b(r,p_r)& =
\frac{[\Lambda_{1c}+r^n\Lambda_{2c}][(\Lambda_{1c}''+r^n \Lambda_{2c}'')p_r^2+2r^m U_c'']}{2r^2}.
 \end{split}
\end{equation*}
To simplify our further analysis, we perform the change of the independent variable
$
t\rightarrow \tau=r(t)^n.
$
This change of  variable, together with the transformation of the derivatives
\begin{equation}
    \frac{\mathrm{d}}{\mathrm{d}t}=\dot \tau\frac{\mathrm{d}}{\mathrm d \tau},\qquad \frac{\mathrm{d}^2}{\mathrm{d}t^2}=\left(\dot \tau^2\right)\frac{\mathrm{d}^2}{\mathrm{d}\tau^2}+\ddot \tau\frac{\mathrm{d}}{\mathrm{d}\tau},
\end{equation}
convert Eq.~\eqref{Eq9} into a linear one
\begin{equation}
\label{eq:rationalized}
\frac{\rmd^2}{\rmd \tau^2}\Theta+p(\tau)\frac{\rmd}{\rmd \tau}\Theta+q(\tau)\Theta=0,
\end{equation}
with the coefficients
\begin{equation*}
\begin{split}
   p(\tau)&= \frac{1}{2n}\left(\frac{2+m+2n}{\tau}-\frac{n\Lambda_{2c}}{\Lambda_{1c}+\Lambda_{2c}\tau}+\frac{hm}{( U_c\tau^\frac{m}{n}-h)\tau}\right),\\
   q(\tau)&=\frac{1}{2n^2\tau^2}\left(\frac{\Lambda_{1c}\Lambda_{2c}''-\Lambda_{2c}\Lambda_{1c}''}{\Lambda_{1c}(\Lambda_{1c}+\Lambda_{2c}\tau)}\tau+\frac{\Lambda_{1c}''}{\Lambda_{1c}}+\frac{U_c''\tau^\frac{m}{n}}{h-U_c\tau^{\frac{m}{n}}}\right).
   \end{split}
\end{equation*}
As we can notice, one of the singular points of Eq.~\eqref{eq:rationalized} depends on the choice of the energy $h$. However, to effectively analyze the differential Galois group of this equation with arbitrary values of the parameters $m,n$, we chose the zero level of the energy $h=0$.

By making the additional rescalation $z=-(A_{2c}/A_{1c})\tau$, we transform Eq.~\eqref{eq:rationalized} into the rational form
\begin{equation}
\begin{split}
&\frac{\rmd^2\Theta }{\rmd z^2}=-\left(\frac{1}{2(1-z)}+\frac{2n+m+2}{2nz}\right)\frac{\rmd \Theta}{\rmd z}-\\ &\left(\frac{4n(\chi_1^2-\chi_2^2)-2m+n-4}{16nz(z-1)}
+\frac{(m+2)^2-4n^2\chi_1^2}{16n^2z^2}\right)\Theta,\label{NVE}
\end{split}
\end{equation}
where 
$\chi_1$ and $\chi_2$ are the new parameters defined by
\begin{eqnarray*}
\begin{split}
\chi_1&=\frac{1}{2n}\sqrt{(m+2)^2+8\left(\frac{U_c''}{U_c}-\frac{\Lambda_{1c}''}{\Lambda_{1c}}\right)},\\
\chi_2&=\frac{1}{2n}\sqrt{(m-n+2)^2+8\left(\frac{U_c''}{U_c}-\frac{\Lambda_{2c}''}{\Lambda_{2c}}\right)}.\label{IC}
\end{split}
\end{eqnarray*}

We immediately recognize that Eq.~\eqref{NVE}  is the Riemann $P$-equation
\begin{equation}
\begin{split}
    \label{eq:Riemann_P_equation}
&\frac{\rmd^2 \Theta}{\rmd z^2}+\left(\frac{1-\gamma-\acute{\gamma}}{z-1}+\frac{1-\alpha-\acute{\alpha}}{z}\right)\frac{\rmd \Theta}{\rmd z}\\ &+\left(\frac{\gamma\acute{\gamma}}{(z-1)^2} +\frac{\beta\acute{\beta}-\alpha\acute{\alpha}-\gamma\acute{\gamma}}{z(z-1)}+\frac{\alpha\acute{\alpha}}{z^2}\right)\Theta=0,
\end{split}
\end{equation}
where $(\gamma, \acute{\gamma}), (\alpha, \acute{\alpha})$ and $(\beta, \acute{\beta})$ refer the exponents at the singular points $z\in\{0, 1, \infty\}$ that must  meet the Fuchs relationship
$
\alpha+\acute{\alpha}+\beta+\acute{\beta}+\gamma+\acute{\gamma}=1.\label{A2}
$
For details, please consult~\cite{Ref5}. We need to calculate the difference of the exponents, and for Eq.~\eqref{NVE}, they take the values
\begin{equation*}
\rho=\alpha-\alpha^\prime=\chi_1,\quad \sigma=\gamma-\gamma^\prime=\frac{3}{2},\quad \tau=\beta-\beta^\prime=\chi_2.\label{Eq12}
\end{equation*}
According to Morales-Ramis theorem, the integrability of  Hamiltonian \eqref{Eq1} implies the identity component of normal variational equation~\eqref{Eq8} and thus also its rationalized form \eqref{NVE} is abelian. Consequently, the corresponding differential Galois group is, in particular, solvable. The necessary and sufficient conditions for the solvability of the identity component of the differential Galois group for the normal variational equation~\eqref{NVE}, that is, Riemann $P$-equation, are formulated as some restrictions. They are given by theorem due to Kimura~\cite{Ref6}.
\begin{theorem}[Kimura]
\label{TH_Kimura} The differential Galois group of the Riemann $P$ equation \eqref{eq:Riemann_P_equation} has a solvable identity component iff\newline
\textbf{A.} at least one of the four numbers $\rho+\sigma+\tau$, $-\rho+\sigma+\tau$, $\rho+\sigma-\tau$, $\rho-\sigma+\tau$ is an odd integer, or\newline
\textbf{B.} the numbers $\rho$ or $-\rho$ and $\sigma$ or $-\sigma$ and $\rho $ or $-\rho$ belongs (in an arbitrary order) to some appropriate fifteen families forming the so-called Schwarz's table, see Table \ref{tablel:Schwarz}.
\end{theorem}

\begin{table}[t]
\setlength{\tabcolsep}{8pt}
\renewcommand{\arraystretch}{1.}
\begin{center}
\begin{tabular}{ccccc}
\hline
1. & $1/2+r$ & $1/2+s$ &  $\mathbb{C}$ &  \\
2. & $1/2+r$ & $1/3+s$ & $1/3+p$ &  \\
3. & $2/3+r$ & $1/3+s$ & $1/3+p$ & $r+s+p$ even \\
4. & $1/2+r$ & $1/3+s$ & $1/4+p$ &  \\
5. & $2/3+r$ & $1/4+s$ & $1/4+p$ & $r+s+p$ even \\
6. & $1/2+r$ & $1/3+s$ & $1/5+p$ &  \\
7. & $2/5+r$ & $1/3+s$ & $1/4+p$ & $r+s+p$ even \\
8. & $2/3+r$ & $1/5+s$ & $1/5+p$ & $r+s+p$ even \\
9. & $1/2+r$ & $2/5+s$ & $1/5+p$ &  \\
10. & $3/5+r$ & $1/3+s$ & $1/5+p$ & $r+s+p$ even \\
11. & $2/5+r$ & $2/5+s$ & $2/5+p$ & $r+s+p$ even \\
12. & $2/3+r$ & $1/3+s$ & $1/5+p$ & $r+s+p$ even \\
13. & $4/5+r$ & $1/5+s$ & $2/5+p$ & $r+s+p$ even \\
14. & $1/2+r$ & $2/5+s$ & $1/3+p$ &  \\
15. & $3/5+r$ & $2/5+s$ & $1/3+p$ & $r+s+p$ even \\
\hline
\end{tabular}
\caption{The Schwarz table. Here $s, p, r \in \mathbb{Z}$.}
\label{tablel:Schwarz}
\end{center}
\end{table} 
\par Now we  apply Theorem~\ref{TH_Kimura} to obtain the necessary integrability  conditions for curved Hamiltonian~\eqref{Eq1}.
\begin{enumerate}
  \item Condition (A) in Theorem \ref{TH_Kimura} is satisfied if at least one of the following numbers
 \begin{equation}
 \begin{split}
  \nu_1&=\chi_1+\chi_2+\frac{3}{2},\qquad \nu_2=-\chi_1+\chi_2+\frac{3}{2},\notag\\
  \nu_3&=\chi_1-\chi_2+\frac{3}{2},\qquad \nu_4=\chi_1+\chi_2-\frac{3}{2},\label{Eq13}
  \end{split}
  \end{equation}
  is an odd integer. This is verified if $\chi_1$ is an arbitrary parameter and $\chi_2$ equals either $\pm\chi_1+2p-\frac{1}{2}$ or $\pm\chi_1+2p+\frac{1}{2}$, where $p\in\mathbb{Z}$.
  \item For Theorem \ref{TH_Kimura}, the only working cases in Table \ref{tablel:Schwarz} for condition (B) are 1, 2, 4, 6, 9, and 14 due to $\sigma=\frac{3}{2}$. We will examine each case individually.  
  
      \textbf{Case 1:} There are two subcases:
      \begin{itemize}
        \item The first subcase is the selection $\pm\rho=\frac{1}{2}+s, s\in \mathbb{Z}$ and $\tau$ is an arbitrary. Consequently, we have $\chi_1=\frac{1}{2}+s$ and $\chi_2$ is arbitrary as a result of $\tau$ is arbitrary.
        \item The second subcase is the choice $\rho$ is arbitrary and $\pm\tau=\frac{1}{2}+p, p\in \mathbb{Z}$. The first condition gives $\chi_1$ is arbitrary and second one implies $\chi_2=\frac{1}{2}+p$.
      \end{itemize}
      \textbf{Case 2:} There are two items:
      \begin{itemize}
        \item The first item is $\pm\tau=\frac{1}{3}+p$ and $\pm\rho=\frac{1}{3}+s$, where $p, s\in\mathbb{Z}$. The first condition gives $\chi_2=\pm\frac{1}{3}+p$ while the second one implies $\chi_1=\pm \frac{1}{3}+s$.
        \item The second item is gives the same conditions as the first item.
      \end{itemize}
      \textbf{Case 4:} There are two choices:
      \begin{itemize}
        \item The first choice is $\pm\rho=\frac{1}{3}+s$ and $\pm\tau=\frac{1}{4}+p$, for some $s, p\in\mathbb{Z}$. The first condition implies $\chi_1=\pm\frac{1}{3}+s$ whereas the second one is satisfied if $\chi_2=\pm\frac{1}{4}+p$.
        \item The second choice is $\pm\rho=\frac{1}{4}+s$ and $\pm\tau=\frac{1}{3}+p$, for some $s, p\in\mathbb{Z}$. These two conditions imply $\chi_1=\pm\frac{1}{4}+ s$ and $\chi_2=\pm\frac{1}{3}+ p$.
      \end{itemize}
      \textbf{Case 6:} There are two selections:
      \begin{itemize}
        \item The first selection is $\pm\rho=\frac{1}{3}+s$ and $\pm\tau=\frac{1}{5}+p$, for some $s, p\in\mathbb{Z}$. These conditions are verified if $\chi_1=\pm\frac{1}{3}+s$ and $\chi_2=\pm\frac{1}{5}+ p$.
        \item The second selection is $\pm\rho=\frac{1}{5}+s$ and $\pm\tau=\frac{1}{3}+p$, for some $s, p\in\mathbb{Z}$. These conditions are satisfied if $\chi_1=\pm\frac{1}{5}+s$ and $\chi_2=\pm\frac{1}{3}+ p$.
      \end{itemize}
     \textbf{Case 9:} There are two sub-cases:
     \begin{itemize}
       \item The first subcase is $\pm\rho=\frac{2}{5}+s$ and $\pm\tau=\frac{1}{5}+p$, for some $s, p\in\mathbb{Z}$. The first condition is satisfied if $\chi_1=\pm\frac{2}{5}+s$ while the second one is verified if $\chi_2=\pm\frac{1}{5}+ p$.
       \item The second subcase is $\pm\rho=\frac{1}{5}+s$ and $\pm\tau=\frac{2}{5}+p$, for some $s, p\in\mathbb{Z}$. These conditions imply $\chi_1=\pm\frac{1}{5}+ s$ and $\chi_2=\pm\frac{2}{5}+ p$.
     \end{itemize}
     \textbf{Case 14:} There are two choices:
     \begin{itemize}
     \item The first choice is $\pm\rho=\frac{2}{5}+s$ and $\pm\tau=\frac{1}{3}+p$, for some $s, p\in\mathbb{Z}$. These conditions give $\chi_1=\pm\frac{2}{5}+ s$ and $\chi_2=\pm\frac{1}{3}+ p$.
     \item The second choice is $\pm\rho=\frac{1}{3}+s$ and $\pm\tau=\frac{2}{5}+p$, for some $s, p\in\mathbb{Z}$. These conditions imply $\chi_1=\pm\frac{1}{3}+ s$ and $\chi_2=\pm\frac{2}{5}+p$.
     \end{itemize}
\end{enumerate}
Thus, we can state the main theorem of this paper as follows:
\begin{theorem}[Main Theorem]\label{TH_MR}
Let $U, \Lambda_{1}$, and $\Lambda_{2}$ are  meromorphic functions of the state variable $\theta$,  and assume there is a point $c\in\mathbb{C}$ such that
\begin{equation}
    \label{eq:condition}
    U_c^\prime=\Lambda_{1c}^{\prime}=\Lambda_{2c}^{\prime}=0,\quad \text{with}\quad U_c\Lambda_{1c}\Lambda_{2c}\neq0.
\end{equation} If the Hamiltonian \eqref{Eq1}  is integrable in the Liouville sense, then the two numbers
\begin{equation}
\begin{split}
\chi_1&=\frac{1}{2n}\sqrt{(m+2)^2+8\left(\frac{U_c''}{U_c}-\frac{\Lambda_{1c}''}{\Lambda_{1c}}\right)},\\
\chi_2&=\frac{1}{2n}\sqrt{(m-n+2)^2+8\left(\frac{U_c''}{U_c}-\frac{\Lambda_{2c}''}{\Lambda_{2c}}\right)},\label{ICC}
\end{split}
\end{equation}
must belong to the integrability Table \ref{tablel:IT}.
\begin{table}[t]
\setlength{\tabcolsep}{7pt}
\renewcommand{\arraystretch}{1.3}
\par
\begin{center}
\begin{tabular}{ccc}
\hline
No. & $\chi_1$ & $\chi_2$ \\ \hline
1. & arbitrary & $\pm\chi_1+2p-\frac{1}{2}$, $\pm\chi_1+2p+\frac{1}{2}$,  $\frac{1}{2} +p$ \\
2. & $\frac{1}{2} +s$ & arbitrary \\
3. & $\pm\frac{1}{3}+ s$ & $\pm\frac{1}{3}+ r$, $\pm \frac{1}{4} +p$, $\pm\frac{1}{5}+r$, $\pm\frac{2}{5}+ p$. \\
4. & $\pm\frac{1}{4}+s$ & $\pm\frac{1}{3}+ r$ \\
5. & $\pm\frac{1}{5}+s$ & $\pm\frac{1}{3}+ p$, $\pm\frac{2}{5}+ r$ \\
6. & $\pm\frac{2}{5}+ s$ & $\pm\frac{1}{5}+ p$, $\pm\frac{1}{3}+r$ \\ \hline
\end{tabular}%
\caption{The integrability table. Here $s, p, r \in \mathbb{Z}$.} \label{tablel:IT}
\end{center}
\end{table}
\end{theorem}
From the form of the integrability Table~\ref{tablel:IT}, we can deduce the following corollary.
\begin{corollary}\label{Prop.NOI}
Let us assume that the quantities
\begin{equation}
\label{eq:qq}
\frac{U_c''}{U_c}\in \mathbb{Q},\qquad \frac{\Lambda_{2c}''}{\Lambda_{2c}}\in \mathbb{Q},\qquad  \frac{\Lambda_{1c}''}{\Lambda_{1c}}\in \mathbb{Q}.
\end{equation}
If   $\chi_1, \,\chi_2$ defined in \eqref{ICC} are irrational, then the Hamiltonian \eqref{Eq1} is not integrable in the sense of Liouville.
\end{corollary}
\begin{proof}
Suppose quantities~\eqref{ICC} have the form
\begin{equation}
\label{eq:chi}
\chi_1=\frac{\sqrt{r}}{2n},\qquad \chi_2=\frac{\sqrt{s}}{2n}, \qquad r,s\in \mathbb{Q},
\end{equation} 
and they are irrational.
If system~\eqref{Eq1} is integrable, then according to  Theorem~\ref{TH_MR} and  integrability Table~\ref{tablel:IT}, one of two possibilities
\begin{equation}
\label{eq:sa}
\sqrt{s}=k+ \sqrt{r},\quad \text{or}\quad \sqrt{s}=k- \sqrt{r},\qquad k\in \mathbb{Z}^*,
 \end{equation} 
  holds true. However, squaring~\eqref{eq:sa}, we get
\begin{equation}
\label{eq:ss}
s=k^2+2k\sqrt{r}+ r, \quad \text{or} \quad  s=k^2-2k\sqrt{r}+ r,
\end{equation} 
which is the contradiction because $r,s\in \mathbb{Q}$, while the right-hand sides of~\eqref{eq:ss} are irrational.  \end{proof}

\section{Applications of Theorem~\ref{TH_MR}}
In this section, we give some examples to clarify the applicability of the obtained result. The first example illustrates that if the Main Theorem~\ref{TH_MR} is not satisfied for a given problem, then this problem is not integrable.  For these models, the metrics have nonzero curvature.
\subsection{Example 1}
As the first example, we consider the integrability of the following  Hamiltonian
\begin{equation}
H=\frac{1}{2}\left[r^n\cos\left(\frac{\theta}{3}\right)+\cos\left(\frac{\theta}{4}\right)\right]\left(p_r^2+\frac{p_{\theta}^2}{r^2}\right)-
r^m\cos\theta,\label{EXA1}
\end{equation}
where $n$ and $m$ are integers. In physics, the Hamiltonian~\eqref{EXA1} describes the motion of a particle due to the influence of potential forces derived from the potential function $V(r,\theta)=-r^m\cos\theta$ on a surface with a non-constant Gaussian curvature.
To gain a quick insight into the dynamics of this model, we computed Poincar\'e cross sections for two pairs of exemplary parameter values, $m$ and $n$. As we can observe, the complex dynamics of the system depicted in Fig.~\ref{fig:ciecia2} suggest its non-integrability. However, these plots were generated for only certain parameter values. For other sets, the dynamics can vary significantly, and in particular, the system may be integrable. Thus, the natural question is whether system~\eqref{EXA1} is integrable for specific values of parameters $m,n$. The answer is no. This example will illustrate the application of the derived integrability obstructions.

\begin{figure}[htp]
	\centering
    \includegraphics[width=0.4\textwidth]{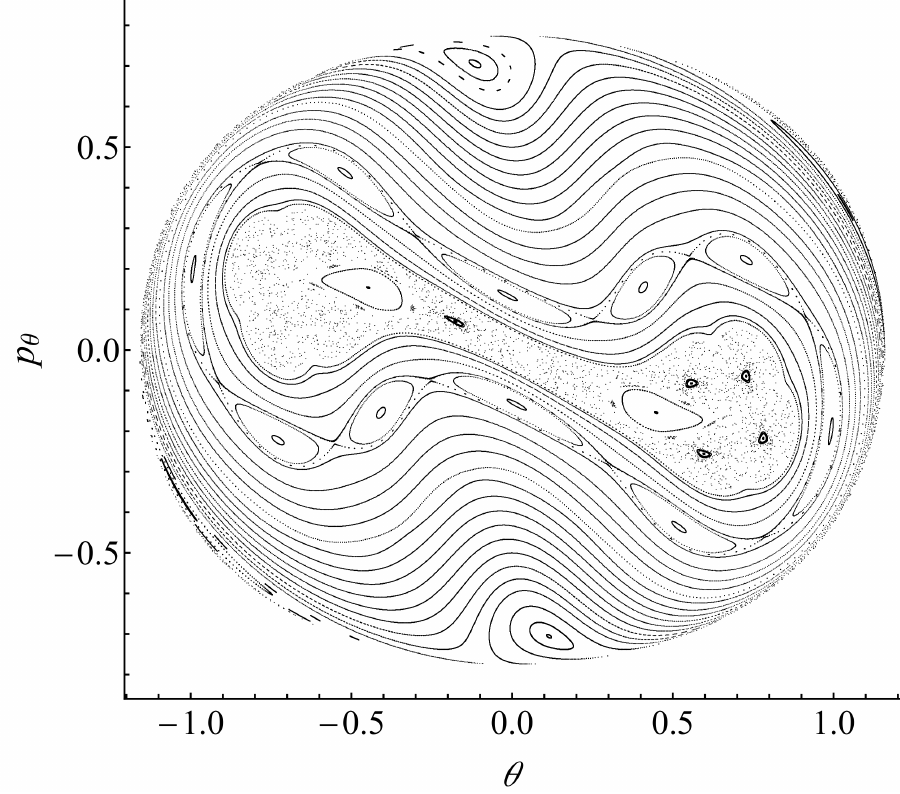} \hspace{1cm}
    \includegraphics[width=0.4\textwidth]{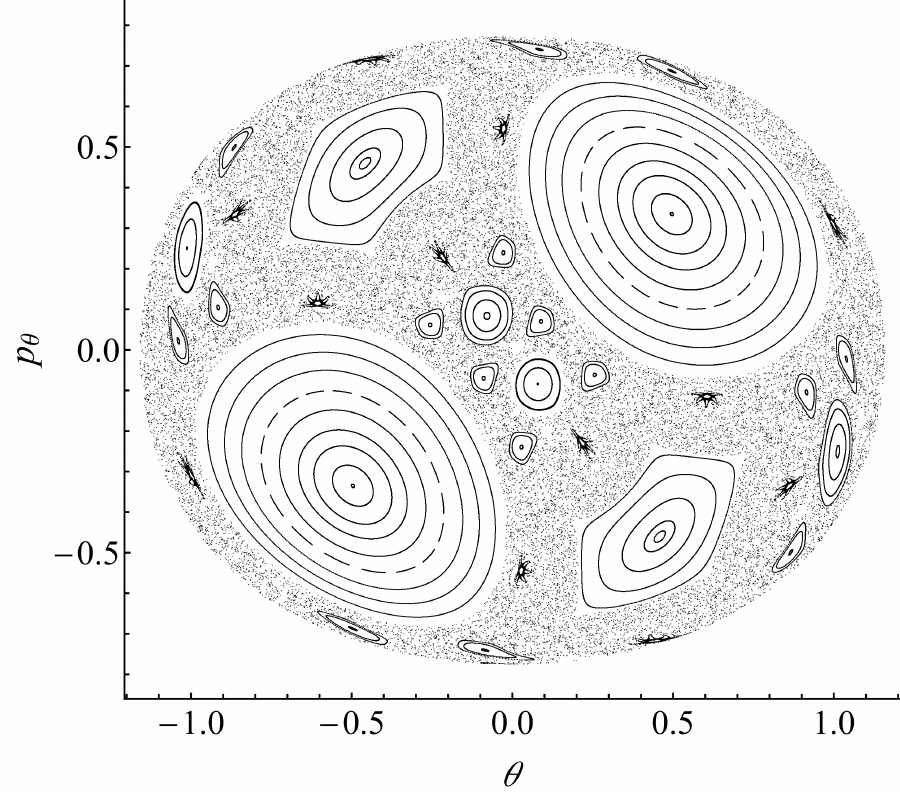}
	\caption{Poincar\'e sections of system~\eqref{EXA1} on surface $r=1$ with $p_r>0$. Top: $m=-1, n=-3$, at the level $E=-0.4$; Bottom: $m=-2, n=-6$, at the level $E=-0.4$ \label{fig:ciecia2}}
\end{figure}

To examine the integrability of the Hamiltonian \eqref{EXA1}, we identify the two Hamiltonians \eqref{Eq1} and~\eqref{EXA1}. We get
\begin{equation*}
\Lambda_1(\theta)=\cos\left(\frac{\theta}{3}\right),\quad \Lambda_0(\theta)=\cos\left(\frac{\theta}{4}\right),\quad U(\theta)=-\cos\theta.\label{EXA3}
\end{equation*}
We take a point, $c=0$, at which the condition~\eqref{eq:condition} is satisfied. At this point,  the integrability coefficients $\chi_1$ and $\chi_2$ are given by
\begin{equation*}
\chi_1=\frac{\sqrt{(m+n)^2-\frac{15}{2}}}{2n},\quad \chi_2=\frac{\sqrt{(m-n+2)^2-\frac{64}{9}}}{2n}.\label{EXA4}
\end{equation*}
It is clear that the two numbers $\chi_1$ and $ \chi_2$ are both irrational for any $m,n\in \mathbb{Z}, n \neq 0$. Therefore, according to Theorem~\ref{TH_MR} and Corollary~\ref{Prop.NOI}, the Hamiltonian~\eqref{EXA1} is not integrable in the class of functions meromorphic in the coordinates $(r,\theta)$ and the momenta $(p_r,p_\theta)$ for all values of the parameters $m,n$, provided $n\neq 0$.

\subsection{Example 2}  As the second example, let us study the following  Hamiltonian system
\begin{equation}
\label{eq:example 3}
    H=\frac{1}{2}\left(\alpha_1\frac{1}{r^2} U(\theta)+\alpha_0\right)\left(p_r^2+\frac{1}{r^2}p_\theta^2\right)+\frac{1}{r^2}U(\theta),
\end{equation}
with
\begin{equation}
\label{eq:uu}
U(\theta)=\frac{1}{\mu_1+\cos(2\theta)+\mu_2\sin(2\theta)},
\end{equation}
where $\alpha_1,\alpha_0,\mu_1,\mu_2$ are arbitrary constants. 
We choose point $c=-\frac{1}{2}\mathrm{i}\ln\left[(-1-\mathrm{i}\mu)/(\sqrt{1+\mu_2^2})\right]$ at which condition~\eqref{eq:condition} is fulfilled.  Taking this into account, we have the following
\begin{equation*}
    n=m=-2,\quad \chi_1=\frac{\sqrt{2}}{\sqrt{1-\mu_1(1+\mu_2^2)^{-\frac{1}{2}}}},\quad \chi_2=-\frac{1}{2}.
\end{equation*}
From the form of the integrability coefficient $\chi_2$, we immediately see that values $\chi_1,\chi_2$  belong to the first item of the integrability Table~\ref{tablel:IT}. Thus, the necessary integrability condition is satisfied for arbitrary values of the remaining constants $\alpha_1,\alpha_0,\mu_1,\mu_2$.  Indeed, the Hamiltonian system defined by Hamiltonian~\eqref{eq:example 3} is integrable with a first integral quadratic in the momenta. Its form is as follows
\begin{equation}
\begin{split}
 I=& \frac{\alpha_1}{r^2 U^3}\left[(3U'^2-2U U'')p_\theta^2-r^2U'^2p_r^2-4rUU'p_\theta p_r\right]+\\ &\frac{\alpha_0}{U^4}\left(4U^2+3U'^2-2UU''\right)p_\theta^2-2\frac{4U^2+U'^2}{U^3}.
 \end{split}
\end{equation}
This integrable case is new. Notice, that if $\alpha_1=0$, so $\Lambda_1$ defined in \eqref{Eq1} vanishes,  the system~\eqref{eq:example 3} becomes separable.

\subsection{Example 3} Finally, we examine the integrability of the more complicated Hamiltonian
\begin{equation}
\begin{split}
H=&\frac{1}{2}\left(\frac{4}{r^4(3+\cos 4\theta)}+\sin\theta\cos\theta\right)\left(p_r^2+\frac{1}{r^2}p_\theta^2\right)\\ &-\frac{1}{8} r^4(3+\cos 4\theta)\sin 2\theta,\label{EX41}
\end{split}
\end{equation}
Identifying two Hamiltonian functions \eqref{Eq1} and \eqref{EX41}, we obtain
\begin{equation}
\begin{split}
& m=4,\quad n=-4,  \quad \Lambda_1=\sin\theta\cos\theta, \\ &\Lambda_2=\frac{4}{3+\cos 4\theta},\quad U=-\frac{1}{8}(3+\cos 4\theta)\sin 2\theta.\label{EX42}
\end{split}
\end{equation}
Now, we will apply the main integrability Theorem~\ref{TH_MR}. The condition \eqref{eq:condition} is satisfied when $c=\pi/4$, and the two numbers \eqref{ICC} become
$
    \chi_1=-5/4,\ \chi_2=-7/4.\label{EX43}
$
The quantities $\chi_1, \chi_2$ belong to the first row in the integrability Table \ref{tablel:IT}. Thus, the Hamiltonian \eqref{EX41} is suspected to be integrable. To confirm the integrability, we should present the complementary integral. Indeed, the system is integrable and the seeking first integral is a quartic polynomial in the momenta
\begin{equation}\label{eq:I}
\begin{split}
&I=\cos4\theta p_r^4-\frac{\sin4\theta}{r}p_r^3p_\theta -  \frac{r^4}{4}\left[\cos8\theta+10\cos4\theta+5\right]p_r^2\\ &+ \frac{r^3}{8}\left[\sin8\theta+14\sin4\theta\right]p_\theta p_r+\frac{r^8}{64}\left[15\cos2\theta+\cos6\theta\right]^2.
\end{split}
\end{equation}
This integrable system is new. The Lie derivative of the integral $I$ along the vector field associated with Hamiltonian \eqref{EX41}, is given by
\begin{equation}\label{eq:ID}
\begin{split}
    \{I,H\}&=\left(h-1\right)W,
    \end{split}
\end{equation}
where $W$ is a non-zero, real function of the state variables.
As outlined by equation~\eqref{eq:ID}, the complementary integral~\eqref{eq:I} is valid only on a fixed energy level, i.e., when  $h = 1$. Consequently, this problem is conditionally integrable. In physics, the Hamiltonian~\eqref{EX41} describes the motion of a particle under the influence of potential forces derived from the potential function  $V(r,\theta) = -\frac{1}{8} r^4 (3 + \cos 4\theta) \sin 2\theta$. This motion occurs on a surface with variable Gaussian curvature, which is determined by substituting~\eqref{EX42} into equation \eqref{Gauss}.

\section{Conclusions}
This paper presents a study of the integrability of certain types of Hamiltonian systems in curved spaces. Thanks to the existence of particular solutions, we were able, with the help of Morales-Ramis theory, to derive necessary integrability conditions. These conditions are expressed as arithmetic constraints on the values of parameters describing the system. It has been demonstrated that the obtained integrability conditions are both effective and straightforward to apply. They were employed in three non-trivial forms of potential and metric, leading to the discovery of certain new integrable cases.

Thus, the conclusions of this paper confirm that Morales-Ramis theory is an effective tool for analyzing the integrability of Hamiltonian systems in curved spaces, with the potential to uncover new and significant integrable cases. These findings expand our understanding of nonlinear dynamics and may have applications in various fields of physics and mathematics. While the current study focused on $2D$ systems, the approach could be extended to higher-dimensional systems, albeit with increased complexity. Our future research will focus on analyzing the integrability of higher-dimensional Hamiltonian systems in curved spaces, which is currently attracting significant scientific interest.

\acknowledgments
This research was funded by The
	National Science Center of Poland under Grant No.
	2020/39/D/ST1/01632. 
 For Open
	 Access, the authors have applied a CC-BY public
	copyright license to any Author Accepted Manuscript~(AAM) version arising from this submission.

\end{document}